\documentclass[11pt]{article}
\usepackage{graphicx,psfrag}
\usepackage[utf8]{inputenc}

\usepackage{amssymb,amsmath,amsthm,enumerate}
\usepackage{fullpage}
\RequirePackage[colorlinks,citecolor=blue,urlcolor=blue]{hyperref}  
\newtheorem{lemma}{Lemma}
\newtheorem{definition}{Definition}
\newtheorem{remark}[lemma]{Remark}

\newtheorem{theorem}[lemma]{Theorem}

\newtheorem{corollary}[lemma]{Corollary}
\newtheorem*{Proposition*}{Proposition}

\usepackage{microtype}
\usepackage{color}

\newcommand{\dE}{\mathbb {E}}
\newcommand{\dP}{\mathbb{P}}

\newcommand{\setR}{\mathbb {R}}

\newcommand{\cE}{\mathcal {E}}

\newcommand{\1}{1\!\!{\sf I}}

\newcommand{\II}{\1}

\renewcommand{\log}{\ln}

\title{Rapid Mixing of Local  Graph Dynamics}

\author{Laurent Massoulié and Rémi Varloot\\ Inria, MSR-Inria Joint Centre} 

\begin{document}

\maketitle

\begin{abstract}
Graph dynamics arise naturally in many contexts. For instance in peer-to-peer networks, a participating peer may replace an existing connection with one neighbour by a new connection with a neighbour's neighbour. Several such local rewiring rules have been proposed to ensure that peer-to-peer networks achieve good connectivity properties (e.g. high expansion) in equilibrium. However it has remained an open question whether there existed such rules that also led to fast convergence to equilibrium.
In this work we provide an affirmative answer: We exhibit a local rewiring rule that converges to equilibrium after each participating node has undergone only a number of rewirings that is poly-logarithmic in the system size. The proof involves consideration of the whole isoperimetric profile of the graph, and may be of independent interest.

 \end{abstract}

\section{Introduction}

With the growing interest for real-world networks, the study of graph dynamics has attracted massive attention. This is especially true in domains such as social networks, the Internet of Things, or wireless sensor networks, which are characterized by ever-shifting topologies. Whereas a lot of attention has been focused on asymptotic properties such dynamics could entail, like degree distribution, connectivity, density and more, another rising trend is the study of the transition period in itself. One notable metric of interest in this approach is the convergence rate: how long does it take for such dynamics to reach a stationary regime?

In this paper, the question is addressed as follows: consider a regular graph with poor connectivity. Introduce a \emph{local dynamic}, which modifies the edge set within a finite neighborhood at each iteration. Can such a process reach stationarity in no more than a polylogarithmic number of modifications per vertex?

\subsection{Related Work}
The properties of graphs resembling real-life networks have been thoroughly studied in works such as \cite{Durrett:07}. One of the more notorious models that has aroused from such studies is the Barab\'asi-Albert model for graphs with power law degree distribution \cite{Albert_statisticalmechanics}. It has also been shown that regular random graphs make good expanders, which yield convenient features, such as the small-world property \cite{KimMedard,Shi99modelsof}.

The dynamics of graphs themselves have been studied in papers such as \cite{bhamidi2011mixing}, which focuses on the convergence of exponential random graphs, or \cite{Papagelis2015refining}, which looks at how to design such dynamics so as to alter the overall graphs structure in a given way. Another concrete application which have received notable attention is the construction p2p networks  \cite{cooper:inria-00070627, Cooper:2009:FMC:1582716.1582742, Feder:2006:LSM:1170136.1170489}.

Regarding local dynamics, different approaches have emerged. \cite{schweinsberg2002bound} delves into the matter of mixing cladograms, for instant, and \cite{jerrum1989approximating} looks at matchings in bipartite graphs. Finally, the best known result concerning the convergence rate of a family of expanders via local edge modifications is given in \cite{zhu2015expanders}, which establishes a mixing time of $O(N^2d^2\sqrt{\log N})$ for $d$-regular graphs over $N$ vertices with $d$ of order $O(\log N)$, that is to say a quasi-linear number of updates per vertex.

\subsection{Our Contribution}
For the remainder of this paper, we consider graphs over vertex set $[N] = \{1,\ldots,N\}$, where $N$ is a positive integer. All asymptotic results will be in respect to $N$. For simplicity, we write $\hbox{polylog}(N)$ to designate $O(\log^k N)$ for some constant positive integer $k$.

Consider the following setting. The vertices in $[N]$ are connected by edges of distinct types. First, a fixed set of edges $(i,i+1)$ constitutes a cycle ($N+1\equiv 1$). Second, each node $n\in [N]$ maintains two pointers, one red and one blue, with respective destinations $b_n$, $r_n$ in $[N]$. These are such that each node $n$ is the destination of only one blue pointer and only one red pointer. It is assumed that an edge $(n,b_n)$ corresponding to such a pointer enables both nodes to communicate directly. 

The continuous-time dynamic then proceeds as follows. The graph evolves through alternating phases. During each phase, only the pointers of a given color evolve, while pointers of the other color are kept fixed.
For the blue phases, blue pointers move along the symmetric graph $G$ formed by the union of the cyclic edges and the unoriented edges $(n,r_n)$ formed by the red shortcuts. Note that $G$ is a 4-regular multigraph. For the red phases, the roles of blue and red pointers are swapped.

Formally, the dynamic is defined as follows. For each edge $e=(i,j)$ of $G$, the two nodes $n,m\in [N]$ such that $b_n=i$ and $b_m=j$ will swap their pointers at the expiration of a random timer whose duration is exponentially distributed with mean 1. These timers are independent across all edges, and are reset independently of everything else upon expiration. Such a process has been studied in the literature, where it is known as the interchange process. See for instance Jonasson \cite{Jonasson_2012} or N. Berestycki \cite{berestycki}, where the discrete time version  of this process is analyzed.

Our main result is then as follows:
\begin{theorem}\label{theorem_1}
	Let $T=\ln(N)^a$ where $a>8$ is a constant. Then after $O(\ln(N))$ phases of length $T$, the graph is, with high probability, distributed so that the sets of blue and red pointers constitute permutations of $[N]$ uniformly and independently distributed over the symmetric group $S_N$.
\end{theorem}

\begin{corollary}
	The above-described process produces with high probability an expander after each node has undergone a number of local connectivity modifications that is polylogarithmic in $N$.
\end{corollary}
\begin{proof}
	By time $\tau:=TO(\ln(N))=O(\ln(N)^{a+1})=\hbox{polylog}(N)$, a given node $n\in [N]$ has seen under these dynamics a number of connectivity modifications that is at most a Poisson random variable with mean $8\tau$. Indeed, at any given time, the rate at which a transition might occur is 8, corresponding to the rate at which the pointer $b_n$ (for the blue phase) issued from $n$ moves (equal to the number of edges of $G$ adjacent to $b_n$, i.e. 4) plus the rate at which the blue pointer ending at $n$ moves (also 4). 
	
	The probability that this number $M_n$ of connectivity modifications exceeds $16\tau$ is then, by Chernoff's bound for deviations of Poisson random variables from their mean, bounded by
	$$
	\dP(M_n\ge 16\tau)\le e^{-8\tau h(16\tau/(8\tau))}=e^{-8\tau h(2)},
	$$
	where $h(x):=x\ln(x)-x+1$ is the Cramér transform of a unit mean Poisson random variable. Since $\tau$ is at least of order $\log^a N$ with $a > 8$, the last term is $o(1/N)$. Thus the probability that at least one node $n\in [N]$ undergoes more than $16 \tau$ local modifications by time $\tau$ is, by the union bound, no more than $N o(1/N)=o(1)$. 
	
	The fact that the resulting graph is an expander is given as part of the proof of Theorem~\ref{theorem_1}, though the fact that such constructions form expanders is a classical result in itself \cite{Bollobas98a, Friedman:1989:SER:73007.73063}.
\end{proof}

\section{Proof strategy}

To proceed, we first introduce some definitions.
\begin{definition}
 For each $k\in[N/2]$, the $k$-th isoperimetric constant $\phi_k(G)$ of a graph $G$ with vertex set $V(G)=[N]$ is
\begin{equation}
\phi_k(G):=\min_{S\subset [N],|S|\le k}|E_G(S,\overline{S})|,
\end{equation}
where $\overline{S}$ denotes the complement $[N]\setminus S$ of a set $S$, $E_G(S,\overline{S})$ denotes the set of edges in $G$ between $S$ and its complement, and $|\cdot |$ denotes the cardinality of a set. 

The collection $\{\phi_k(G)\}_{k\in[N/2]}$ of isoperimetric constants of graph $G$ constitutes its {\bf isoperimetric profile}. 
\end{definition}
We shall omit the argument $G$ in these quantities when it is clear from context.

Our  proof strategy is then as follows.

We shall control the evolution of the  isoperimetric profile of the graph along which pointers move from one phase to the next, establishing lower bounds on this profile in an iterative manner. Specifically, we show the following
\begin{lemma}\label{lemma_1}
Let $\beta > 1$ be a constant such that $a> 2\beta+1$, and $\gamma = \log^{-\beta} N$, and let $d$ be an integer. Suppose that, at the end of a red phase, the graph $G$ consisting of the red edges and the ring is such that, for all $S\subset [N]$ with $|S|\le N/2$, 
$$
|E(S,\bar{S})| \geq \min(\gamma |S|, d),
$$ 
or in other words, that $\phi_k(G)\ge \min(\gamma, d/k)$ for all $k\le N/2$.

Then at the end of the following blue phase, with probability at least $1-o(1/N)$, the graph $G'$ consisting of the blue edges and the ring is such that, for all $S$ in $[N]$ with $|S|\le N/2$, 
$$|E_{G'}(S,\bar{S})| \geq \min(\gamma |S|, 2d),$$
or in other words, $\phi_k(G')\ge \min(\gamma, 2d/k)$ for all $k\le N/2$.
\end{lemma}

To prove this, we first show a stronger lower bound on the average $\dE |E_{G'}(S,\overline{S})|$, namely we establish the following

\begin{lemma}
\label{lemma_2}
Under the assumptions of Lemma~\ref{lemma_1}, for all $S\subset [N]$ with $|S|\le N/2$, denoting by $F(S,\overline{S})$ the set of blue edges at the end of the blue phase between $S$ and $\overline{S}$,  it holds that 
\begin{equation}
\dE |F(S,\overline{S})|\ge \frac{1}{2\gamma} \min(\gamma|S|,2d). 
\end{equation}
\end{lemma}
Lemma~\ref{lemma_1} is then deduced from Lemma~\ref{lemma_2} by invoking some concentration inequalities together with union bounds. Details are given in Section~\ref{sec:fluctuations}.

An easy consequence of Lemma~\ref{lemma_1} is the following
\begin{corollary}\label{corollary_2}
After $\hbox{log}_2(N)$ phases, with high probability the graph $G$ on which pointers evolves admits the following lower bound on its isoperimetric profile:
\begin{equation}
\phi_k(G)\ge \gamma,\;\; k\le N/2.
\end{equation}
\end{corollary}
\begin{proof}
Clearly, at the beginning of the first phase the assumptions of Lemma~\ref{lemma_1} are satisfied with $d=2$. Indeed, any subset $S\subset [N]$ of size $|S|\le N/2$ is connected by at least two edges (that come from the cycle) to its complement $\overline{S}$, so that 
$$
|E_G(S,\overline{S})|\ge 2\ge \min (\gamma |S|,2).
$$
Denote by $\cE_f$ the event that at the beginning of phase $f$, the graph $G_f$ on which pointers will evolve in the next phase satisfies the assumptions of Lemma~\ref{lemma_1} with parameter $d_f=2^{f+1}$. Thus we have just established that event $\cE_0$ holds with certainty, and Lemma~\ref{lemma_1} entails that 
$$
\dP(\overline{\cE}_{f+1}|\cE_f)\le o(1/N),\; f\ge 0.
$$
Thus
$$
\begin{array}{ll}
\dP(\overline{\cE}_{f+1})&=\dP(\overline{\cE}_{f+1}|\cE_f)\dP(\cE_f)+\dP(\overline{\cE}_{f+1}|\overline{\cE_f})\dP(\overline{\cE_f})\\
&\le
o(1/N)+\dP(\overline{\cE}_{f}).
\end{array}
$$
By induction on $f$, this yields 
$$
\dP(\overline{\cE}_{f})\le o(f/N).
$$
For $f=\hbox{log}_2(N)$, the right-hand side of this expression is $o(1)$, so that with high probability, after $\hbox{log}_2(N)$ phases, the graph $G$ on which pointers evolve verifies for all $k\le N/2$:
$$
\phi_k\ge \min(\gamma, 2^{f+1}/k)\ge \min(\gamma, N/k)=\gamma.
$$
\end{proof}
The proof of Theorem~\ref{theorem_1} is then concluded as follows:
\begin{proof}
By Corollary~\ref{corollary_2}, after $f=\hbox{log}_2(N)$ phases, the graph on which pointers evolve is a $\gamma$-expander, i.e. its isoperimetric constant $\phi_{N/2}$ is at least $\gamma$. We shall exploit this property to show that the interchange process on $G$ mixes in time $\hbox{polylog}(N)\le T$, so that with high probability, in two more phases our process will have reached stationarity. 

Our main tool to this end is Theorem 4.6, p. 47 in Berestycki \cite{berestycki}, which gives a sufficient condition for the (discrete time version of the) interchange process on a graph $G$ to mix in time $T$. Adapted to our continuous time setup, this theorem guarantees that the continuous time interchange process on $G$, which in time $T$ incurs on average $|E|T$ transitions of the discrete time process, where $|E|$ is the number of edges of $G$, will have mixed in time $T$ provided 
\begin{equation}\label{interchange_mix}
T\ge 8 \ln(N)\Delta K/N.
\end{equation} In this expression, the quantities $\Delta$ and $K$ are defined as follows. For each pair of nodes $(i,j)\in [N]$, one must define a path $\gamma_{ij}$ on $G$ connecting these two nodes. 
$\Delta$ is then defined as the largest length of all paths $\gamma_{ij}$, and $K$ as the supremum over edges $e$ in $G$ of the number of paths $\gamma_{ij}$ crossing $e$. According to Lemma~\ref{expanders_etc} below, for a $\gamma$-expander with constant node degrees of order 1, we can choose these paths such that $\Delta=O(\ln(N)/\gamma^2)$ and $K=O(N\ln(N)^2/\gamma^2)$. Plugged into \eqref{interchange_mix}, these evaluations imply that mixing has occurred by time $T$ provided $T$ is large compared to $\ln(N)^4/\gamma^4$. Since we have assumed $T=\ln(N)^a$ with $a>8$, this follows by our choice of $\gamma=\ln(N)^{\beta}$ where the only constraint on $\beta$ is $\beta>1$. 
\end{proof}

\begin{lemma}\label{expanders_etc}
Let $G$ be a $d$-regular graph on $[N]$ nodes, such that $\phi_{N/2}(G)\ge \gamma$. Then one can construct for each pair of nodes $(i,j)\in [N]$
 paths $\gamma_{ij}$ on $G$ each of length no larger than $\Delta=2\ln(N)d^2/\gamma^2$, and such that each edge $e$ of $G$ is crossed by $O(N\ln^2(N)d^2/\gamma^2)$.
\end{lemma}
\begin{proof}
The celebrated Cheeger's inequality (see e.g. Berestycki \cite{berestycki} Theorem 3.4 p.35) ensures that the spectral gap for the discrete time random walk on a $d$-regular graph $G$ with $\phi_{N/2}(G)\ge \gamma$ is at least $\gamma^2/(2d^2)$. Thus the total variation distance between the distribution of the random walk at time $\Delta:=2 d^2\ln(N)/\gamma^2$ and the uniform distribution on $G$ is $o(1/N)$ (this follows e.g. by Theorem 2.2, p. 22 in \cite{berestycki}). As a result, for any $i\in[N]$, the probability that the walk started at $i$ hits $j$ at time $\Delta$ is at least $1/{2N}$. Consider then the following randomized construction. For each $i$, create $5 N\ln(N)$ independent walks of length $\Delta$ started at $i$. The probability that for some particular $\in[N]$, no such walk  issued from $i$ hits $j$ is then  at most
$$
\left(1-1/2N\right)^{5N\ln(N)}\le e^{-5\ln(N)/2}=o(N^{-2}).
$$ 
Thus with high probability, the collection of paths thus created joins every node $i$ to every node $j$. 

Let us now evaluate the number of times a given edge $e=(u,v)$ of $G$ is traversed by this collection of paths. This is no larger than the number of times these paths visit node $u$. For $t\le 5N\ln(N)$, denote by $X_i(t)$ the number of visits to $u$ by  the $t$-th path sampled with starting point $i$. Clearly, $X_i(t)\le \Delta$. Also, 
$$
\dE \sum_{i\in [N]}\sum_{t\le 5 N\ln(N)}X_i(t)=5N\ln(N)\sum_{i\in [N]}\sum_{\ell=0}^{\Delta}P_{i u}^{(\ell)},
$$
where $P_{i u}^{(\ell)}$ denotes the transition probability from $i$ to $u$ in $\ell$ steps of the walk. However the walk is symmetric, so that $P_{i u}^{(\ell)}=P_{u i}^{(\ell)}$. The above expection thus reads
$$
\dE \sum_{i\in [N]}\sum_{t\le 5 N\ln(N)}X_i(t)=5N\ln(N)(\Delta+1).
$$
Let $Z=\sum_{i\in [N]}\sum_{t\le 5 N\ln(N)}X_i(t)$ denote the total number of visits to $u$ by all paths. For some arbitrary $C>0$, Hoeffding's inequality then gives
$$
\dP(Z\ge \dE(Z)+CN\Delta\ln(N))\le \exp\left(-\frac{C^2N^2\Delta^2\ln(N)^2}{\Delta^2 5 N^2\ln(N)}  \right)=e^{-C^2\ln(N)/5}.
$$
Taking $C=3$ (say), the right-hand side is $o(N^{-1})$. Thus with high probability, no node $u$ is visited more than $(9N\ln(N)\Delta)$ times by the collection of constructed paths. The announced result follows.
\end{proof}

\section{From bounds in expectation to bounds with high probability}

We now give the derivation of Lemma~\ref{lemma_1} from the result of Lemma~\ref{lemma_2}. 

We thus consider the graph $G$ on $[N]$ formed by edge cycles together with shortcut cycles after their evolution through a phase, and thus have by Lemma~\ref{lemma_2} that for each fixed set $S\subset[N]$, the number $|F(S,\overline{S})|$ of such shortcut edges connecting $S$ to $\overline{S}$ verifies
$$
\dE |F(S,\overline{S})|\ge \frac{1}{2\gamma}\min(\gamma |S|, 2d).
$$
Fix $k\le N/2$. We further restrict ourselves to $k\ge 2/\gamma$, since for smaller $k$ one clearly has $\phi_k\ge \gamma$, because of the presence of at least two cycle edges in $E(S,\overline{S})$ in any $S$ of size $k\in[N/2]$. The conclusion of Lemma~\ref{lemma_2} is thus immediate for smaller $k$.

For some set $S$ of size $k$, let $\ell$ be the number of contiguous portions of the cycle it is made of. Thus $\ell\in [k]$, and we have
$$
|E(S,\overline{S})|=|F(S,\overline{S})|+2\ell.
$$
We will need the following two results.
\begin{lemma}
Conditionally on the shortcut configuration at the beginning of the considered phase, the random variable $|F(S,\overline{S})|$ consists of the sum of negatively associated Bernoulli random variables. Consequently, for any $r\in(0,1)$, it holds that
\begin{equation}\label{chernoff_2}
\dP(|F(S,\overline{S})|\le r (2\gamma)^{-1}\min(\gamma |S|, 2d))\le e^{-(2\gamma)^{-1} \min(\gamma |S|, 2d)h(r)},
\end{equation}
where $h(r):=r\ln(r)-r+1$.
\end{lemma}
\begin{proof}
Represent the collection of termination points of pointers through the binary variables $\xi_i\in\{0,1\}$, $i\in[N]$ where $\xi_i=1$ if and only if one pointer issued from $S$ points towards $i$.
The set of variables $\{\xi_i\}$ evolves, under the interchange process dynamics, as a symmetric exclusion process. Moreover, when we condition on the initial configuration, its distribution is deterministic, and therefore satisfies a strong form of negative dependence known as {\em strong Rayleigh property}, see Borcea et al. \cite {borcea}. It then follows from \cite{borcea}, Proposition 5.1, that the collection of indicator variables $\{\xi_i(T)\}$ also satisfies this strong Rayleigh property at time $T$ when the phase is completed. Strong Rayleigh property implies negative association (see again \cite{borcea}, Section 2, Figure 1). It then follows from Dubhashi and Ranjan \cite{dubhashi} that $|F(S,\overline{S})|$, which also reads
$$
|F(S,\overline{S})|=\sum_{i\in\overline{S}}\xi_i(T),
$$
verifies the same Chernoff bounds that it would if the $\xi_i(T)$ were mutually independent. The announced result \eqref{chernoff_2} then follows from classical evaluations of Chernoff bounds.
\end{proof}
\begin{lemma}
The number of sets $S\subset[N]$ of size $k$ that consist of $\ell$ contiguous portions of the cycle is at most $N^{2\ell}$.
\end{lemma}
\begin{proof}
We may enumerate such sets $S$ by scanning the cycle $[N]$ starting from 1, and identifying the first time we find a starting point of an interval in $S$, then the end point of that interval, and so on. Clearly this will produce $2\ell$ numbers in $[N]$, which characterize $S$, hence the result.
\end{proof}
The union bound gives us the following bound on the probability $p_k$ that for some set $S$ of size $k$, one does not have the desired property $|E(S,\overline{S})|\ge \min(\gamma k,2d)$:
$$
p_k\le \sum_{\ell=1}^k N^{2\ell}\dP(|F(S,\overline{S})|\le \min(\gamma k,2d)-2\ell).
$$
We now distinguish according to whether $\gamma k\le 2d$ or not.

{\bf Case 1: $\gamma_k \le 2d$}. We then have
$$
\begin{array}{lll}
p_k&\le& \sum_{\ell=1}^{\gamma k/2}N^{2\ell} \exp\left(-(2\gamma)^{-1}\gamma k h(\frac{\gamma k-2\ell}{(2\gamma)^{-1}\gamma k})\right)\\
&\le &N\exp\left(\gamma k \ln(N) -(2\gamma)^{-1}\gamma kh(o(1))\right)\\
&=&\exp\left(\ln(N)[1+\gamma k -\gamma k (1/2)\ln(N)^{\beta-1}h(o(1))]\right).
\end{array}
$$
The term in square brackets is asymptotically equivalent to $-\gamma k(1/2)\ln(N)^{\beta-1}h(o(1))$, because $h(o(1))$ tends to 1, and we assumed $\beta>1$. Moreover, since $\gamma k\ge 1/2$, the whole exponent is large compared to $\ln(N)$. Thus $p_k=o(N^{-r})$ for any fixed $r>0$.

{\bf Case 2: $\gamma_k > 2d$}. We then have
$$
\begin{array}{lll}
p_k&\le& \sum_{\ell=1}^{d}N^{2\ell} \exp\left(-(2\gamma)^{-1}d h(\frac{2d-2\ell}{(2\gamma)^{-1}2d})\right)\\
&\le &N\exp\left(2d \ln(N) -(2\gamma)^{-1}2d kh(o(1))\right)\\
&=&\exp\left(\ln(N)[1+2d -2d (1/2)\ln(N)^{\beta-1}h(o(1))]\right).
\end{array}
$$
We can then conclude as in the preceeding case.
\label{sec:fluctuations}

\section{Controlling the mean}

The goal of this Section is to prove Lemma~\ref{lemma_2}. We thus assume to be given a graph $G$ on $[N]$, which in our context is constituted of a cycle plus one set of (red, say) pointers, so that $G$ is a 4-regular undirected graph. The structural assumption made on $G$ is that for some integer $d$ and some positive $\gamma$ (in our context, $\gamma=\ln(N)^{-\beta}$ for some fixed $\beta>1$) the isoperimetric profile of $G$ verifies
$$
\phi_k(G)\ge \min(\gamma k,d),\; k\le N/2.
$$
Our goal is to prove that for any fixed set $S$ of size $k\le N/2$, after $T$ time steps, by the end of the next phase, the expected number $\dE|F(S,\overline{S})|$ of (blue, say) pointers connecting $S$ to $\overline{S}$ after $T$ time steps is on average at least $(2\gamma)^{-1}\min(\gamma k, 2d)$.

We shall divide the proof into two parts, arguing differently depending on the size $k$ of considered sets $S$. 

\subsection{Small sets: from partial expansion to partial spread}
To deal with smaller values of $k$, we develop a new tool, which establishes a lower bound on the speed with which the mass of a random walk on a graph is {\em partially} spread, which only involves a single isoperimetric number $\phi_k$ of the graph. Of course, the {\em partial} spread of the mass is related to the corresponding value of $k$. 

The general framework and statement are as follows. The proof is deferred to Section \ref{sec:partial_spread}.

Let $G$ be an undirected graph on $n$ nodes, with maximal degree $\Delta$. We consider $\{X_t\}$,  the continuous time random walk on $G$, 
Our aim is to prove the following
\begin{theorem}\label{theorem_2}
Let $G$ be an undirected graph on node set $[N]$ with maximal degree $\Delta$, and $\{X_t\}$ the continuous time random walk on $G$, i.e. the Markov jump process on the vertex set $[N]$, with as its non-zero jump rates $q_{ij}=1$ for all edges $(i,j)$ of the graph (its infinitesimal generator is then $-L$ where $L$ is the Laplacian matrix of $G$). Let $\{\pi_i(t)\}_{i\in[N]}$ denote its law at time $t$. 

Let $k\le N/2$ be fixed, and define the isoperimetric constant $\phi_k(G)$ as 
$$
\phi_k(G):=\inf_{S\subset [n], |S|\le k}\frac{|E(S,\overline{S})|}{|S|}\cdot
$$
For an arbitrary initial distribution of the  random walk, for any set $S$ such that $|S|\le k$ and any $t\ge 0$, one has:
\begin{equation}\label{partial_spread}
\sum_{i\in S} \pi_i(t)\le \frac{|S|}{k+1}+\sqrt{k+1}e^{-\lambda^*_2 t},
\end{equation}
where 
\begin{equation}\label{lambda_lower_bound}
\lambda^*_2=\frac{\phi_k(G)^2}{2\Delta}\cdot
\end{equation}
\end{theorem}

\begin{remark}
The quantity $\lambda^*_2$ is of the same form as the lower bound on the spectral gap $\lambda_2$ of the Laplacian that the celebrated Cheeger inequality gives when $k=n/2$. In this classical situation, instead of \eqref{partial_spread} one has the conclusion that $d_{var}(\pi(t),\mathcal{U}([n]))\le \sqrt{n}e^{-\lambda^*_2 t}$. 
\end{remark}

Let us now use Theorem~\ref{theorem_2} to prove the conclusion of Lemma~\ref{lemma_2} for small values of $k$. 

Specifically, let $\kappa=4 \gamma^{-1}$, $k_d:=\kappa d$, and consider $k\le k_d$. For a fixed set $S$ of size $k$, and a fixed node $i\in S$, let $X_i(t)$ denote the location of the pointer issued from $i$ at time $t$. Under the dynamics we consider, $X_i(t)$ corresponds to an ordinary random walk on the graph $G$. Moreover, the assumptions of Lemma~\ref{lemma_2} guarantee that the graph $G$ satisfies 
$$
\phi_{3k}(G)\ge \min(\gamma,d/3k)\ge \min(\gamma,1/(3\kappa))=\gamma/(12).
$$
By theorem~\ref{theorem_2}, one therefore has
$$
\dP(X_i(T)\in S)\le \frac{|S|}{3k}+\sqrt{3k+1}e^{-\lambda^*_2 T},
$$
where $\lambda^*_2=\phi_{3k}(G)^2/(2\Delta)\ge \gamma^2/(12^2\cdot 2\cdot 4)$. Recall that $T=\ln(N)^a$ and that $\gamma=\ln(N)^{-\beta}$ for some $\beta>1$. By assumption, $a>3$. For some suitable choice of $\beta$, we then have $a-2\beta>1$ and thus
$$
\dP(X_i(T)\in S)\le \frac{|S|}{3k}+\sqrt{3k+1}e^{-\lambda^*_2 T}\le \frac{1}{3}+\exp\left(-\Omega(\ln(N)^{a-2\beta})\right)\le 1/2.
$$
Summming over $i\in S$, we obtain that the expected number $\dE|F(S,\overline{S})|$ of pointers issued from $S$ that point into $S$ at the end of the phase is no larger than $|S|/2$.

\subsection{Large sets}
We now deal with values of $k$ in the range $[\kappa d, N/2]$. Throughout this section we consider a fixed set $S$ of size $k$.

In this context, we define $\pi_i(t)$ to be $1/k$ times the probability that a pointer issued from $S$ targets $i$, conditional on the initial configuration of these pointers at the beginning of the phase. We also let $\pi_{(i)}(t)$ denote the $i$-th largest value $\pi_j(t)$, $j\in[N]$, and $\pi_{[m]}(t):=\sum_{i\in[m]}\pi_{(i)}(t)$ denote the cumulative mass that the probability distribution $\pi(t)$ puts on the $m$ nodes where its mass is the largest.

Obviously, one has 
$$
\pi_{(i)}(0)=\frac{1}{k}\II_{i\in[k]}.
$$ 
We now establish a property of the time derivative $\frac{d}{dt}\pi_{[m]}(t)$:
\begin{lemma}\label{lem:cutcut}
Under the assumptions of Lemma \ref{lemma_1} that $\phi_m(G)\ge \min(\gamma m, d/m)$, $m\in [N]$, one has the inequalities
\begin{equation}\label{cutcut}
\frac{d}{dt}\pi_{[m]}(t)\le -4 \sum_{j=1}^{d_m}\left(\pi_{(m_j+1)}-\pi_{(m-j+1+d_m)}\right),
\end{equation}
where $d_m=\lfloor \min(\gamma m, d)/4\rfloor$.
\end{lemma}
\begin{proof}
Assume to simplify notation that the permutation which sorts nodes $i$ in $[N]$ in decreasing order of $\pi_i$ is the identity, so that $\pi_i(t)=\pi_{(i)}(t)$. The time derivative of $\pi_{[m]}$ then reads
$$
\frac{d}{dt}\pi_{[m]}(t)=\sum_{i\in [m]}\sum_{j>m}\II_{i\sim j}(\pi_j-\pi_i).
$$
Indeed, changes in the mass $\pi_{[m]}$ result from interchange of pointer extremities $i,j$ with $i\le m$ and $j>m$, which occur at unit rate; when one such interchange occurs, the expected change to $\pi_{[m]}$ is precisely $\pi_j-\pi_i$. Now the number of such edges is by assumption at least $\min(\gamma m, d)$. Moreover, the number of such edges adjacent to any  node is at most 4, because the graph has degree bounded by 4. 

The value of the right-hand side in the above equation, because the $\pi_i$ are sorted in decreasing order, is minimized when the edges crossing the cut between $[m]$ are adjacent to nodes with index closest to $m$. The degree constraint then implies the upper bound \eqref{cutcut}. 
\end{proof}
Let $d':=\lfloor d/4\rfloor$. Let $I:=[k-(2/3)k_d,k+(2/3)k_d]$. We now introduce an auxiliary process $\{\nu_i(t)\}_{i\in[N],t>0}$ defined via:
$$
\begin{array}{lll}
\nu_i(0)&=\frac{1}{k}\II_{i\in[k]},&i\in[N],\\
\frac{d}{dt}\nu_i(t)&=4\II_{i\in I}\left[(\nu_{i-d'}(t)-\nu_i(t))\II_{i-d'\in I}+4(\nu_{i+d'}-\nu_i)\II_{i+d'\in I}\right],&i\in[N],\; t>0.
\end{array}
$$
The probability distribution $\nu(t)$ is readily interpreted as the law at time $t$ of a random walk started with uniform distribution on $
[k]$, that jumps from $i$ to $i+d'$ (resp., $i-d'$) at rate 4, provided both $i$ and the destination $i\pm d'$ lie in $I$. 

Denoting $\nu_{[i]}(t):=\sum_{j\in[i]}\nu_j(t)$, we then have the following
\begin{lemma}
For all $t>0$, $i\in [N]$, it holds that 
$$
\pi_{[i]}(t)\le \nu_{[i]}(t).
$$
\end{lemma}
\begin{proof}
Introduce the differences $\delta_i(t):=\pi_{[i]}(t)-\nu_{[i]}(t)$. It is readily seen that $\delta(0)\equiv 0$.
The arguments in the proof of Lemma \ref{lem:cutcut} readily imply that
$$
\frac{d}{dt}\pi_{[i]}(t)\le -4\sum_{j=1}^{d'}\II_{i-j+1\in I}\II_{i-j+1+d'\in I}\left(\pi_{(i-j+1)}-\pi_{(i-j+1+d')}\right),
$$
and the same equation holds with identity for distribution $\nu(t)$. There thus exist integers $m_i\ge 0$ for all $i\in [N]$ such that $i-m_i\ge 0$, $i+m_i\le N$, and furthermore for all $i\in[N]$,
$$
\begin{array}{ll}
\frac{d}{dt}\pi_{[i]}(t)&\le -4\left( 2\pi_{[i]}-\pi_{[i-m_i]}-\pi_{[i+m_i]}\right),\\
\frac{d}{dt}\nu_{[i]}(t)&=-4\left( 2\nu_{[i]}-\nu_{[i-m_i]}-\nu_{[i+m_i]}\right),\\
\end{array}
$$
so that
$$
\frac{d}{dt}\delta_i\le -4\left(2\delta_i-\delta_{i-m_i}-\delta_{i+m_i}\right).
$$
In the above, boundary conditions are given by $\delta_0=\delta_N=0$. This equation implies that necessarily, the supremum over $i\in[N]$ of $\delta_i$ cannot become positive, because its derivative is always non-positive. 
\end{proof}
By the previous lemma, an upper bound on $\pi_{[k]}(T)$ is provided by $\nu_{[k]}(T)$. However the latter quantity is simpler to analyze. It can be interpreted as $1/k$ times the average number of points of $(2/3)k_d$ random walks initialized at each point in $[k-(2/3)k_d,k]$ which fall within $[k]$ at time $T$. These walks proceed with jumps of size $\pm d'$ at rate 4, constrained to not leave interval $I=[k-(2/3)k_d,k+(2/3)k_d]$. 

For a given initial condition $i\in[k-(2/3)k_d]$, the number of sites it can visit is of the order of $(4/3)k_d/(d')=\Theta(\kappa)=\Theta(\ln(N)^{\beta})$. Recall that we have chosen $T=\ln(N)^a$ with $a>2\beta +1$. Classical results on the nearest neighbor on an interval $[M]$ state that it mixes in time of the order of $M^2$ \cite{peres2009book}. Thus each of the random walks just introduced mix in time $T$. We therefore have the following evaluation:
$$
\pi_{[k]}(T)\le \nu_{[k]}(T)\le 1- \frac{(2/3)k_d}{k}\left(1/2-o(1)  \right).
$$ 
The expected number $\dE|F(S,\overline{S})|$ is then lower-bounded by
$$
\dE|F(S,\overline{S})|\ge (2/3)k_d(1/2-o(1))=(1/3)\kappa d=[1/3-o(1)]4\gamma^{-1}d\ge \frac{1}{2\gamma}(2d).
$$
The announced result follows. 
\section{Proof of Theorem 2}
\label{sec:partial_spread}
\begin{proof}
In vector form the law $\pi(t)$ of the random walk on $G$ at time $t$ reads $\pi(t)=e^{-t L}\pi(0)$, where $L$ is the graph's Laplacian. Its entries $\pi_i(t)$ are thus linear combinations of $n$ functions of the form $e^{-\lambda_j t}$, where $\lambda_j$ are the eigenvalues of $L$, and so is the difference $\pi_i(t)-\pi_j(t)$. It can be shown by induction on $N$ that such linear combinations of $N$ distinct exponential functions are either identically zero in $t$, or admit at most $N-1$ distinct roots in $t$. Thus
 for any $i\ne j$, either $\pi_i(t)\ne \pi_j(t)$ except perhaps for finitely many $t$, or else $\pi_i(t)\equiv \pi_j(t)$ for all $t\ge 0$. 

We can thus split $\setR_+$ into finitely many intervals $I^{(1)}=[0,t_1)$, $I^{(2)}=[t_1,t_2),\ldots$, and on each interval $I^{(j)}$ determine a particular permutation $\sigma^{(j)}$ of $[N]$ such that for all $j$, and all $t\in I_j$, one has
$$
\pi_{\sigma^{(j)}(1)}(t)\ge \pi_{\sigma^{(j)}(2)}(t)\ge \cdots \ge \pi_{\sigma^{(j)}(N)}(t).
$$

For $t$ in any given interval $I^{(j)}$, we will maintain an auxiliary probability distribution on $[k+1]$, denoted $\{\nu_i(t)\}_{i\in[k+1]}$. This distribution can be interpreted as that of a random walk on a graph $G^{(j)}$ with node set $[k+1]$, obtained from $G$ as follows. We identify node $\sigma^{(j)}(i)$ in $G$ with node $i$ in $G^{(j)}$ for all $i\in [k]$, and collapse all nodes $\sigma^{(j)}(u)$, $u>k$ to form node $k+1$. All edges are then preserved, so that the adjacency matrix $A^{(j)}$ of $G^{(j)}$ is given by
$$
\begin{array}{lll}
A^{(j)}_{u,v}&=A_{\sigma^{(j)}(u),\sigma^{(j)}(v)},&u,v\in [k],\\
A^{(j)}_{u,k+1}&=\sum_{v=k+1}^n A_{\sigma^{(j)}(u),\sigma^{(j)}(v)}, &u\in [k],
\end{array}
$$
where $A$ is the adjacency matrix of $G$.
For convenience, we denote by $\pi_{(i)}(t)$ the $i$-th largest entry of distribution $\pi(t)$. Thus for $t\in I^{(j)}$, $\pi_{(i)}(t)=\pi_{\sigma^{(j)}(i)}(t)$.

The result of the theorem will then follow from the combination of two ingredients. We first show in Lemma \ref{lemme_1} below that, for all $t$, one has the following bound:
\begin{equation}\label{eq:1}
\pi_{(i)}(t)\le \nu_{i\wedge (k+1)}(t),\quad i\in[N],\; t\ge 0.
\end{equation}

We then establish in Lemma \ref{lemme_2} below that for all $j$, the second smallest eigenvalue $\lambda^{(j)}_2$ of the Laplacian of $G^{(j)}$ is lower-bounded by $\lambda_2^*$ given in \eqref{lambda_lower_bound}, where crucially $\Delta$ is the largest node degree in $G$, not in $G^{(j)}$. 

This readily implies the $L^2$ control
$$
\sum_{i\in[k+1]}\left|\nu_i(t)-\frac{1}{k+1}\right|^2\le e^{-2\lambda^*_2t}.
$$
Cauchy-Schwarz inequality then gives the following control on variation distance:
$$
\sum_{i\in [k+1]}|\nu_i(t)-1/(k+1)|\le \sqrt{k+1} e^{-\lambda^*_2 t}. 
$$
Together, these two results entail that for all $s\le k$,
\begin{equation}\label{eq:3}
\sum_{i\in [s]}\pi_{(i)}(t)\le \frac{s}{k+1}+\sqrt{k+1} e^{-\lambda^*_2 t},
\end{equation}
which is the announced result.
\end{proof}

\begin{lemma}\label{lemme_1}
The distributions $\pi(t)$, $\nu(t)$ verify bound \eqref{eq:1}.
\end{lemma}
\begin{proof}
The bound trivially holds at $t=0$. We can establish it by induction on each interval $I^{(j)}$. Let us then consider one such interval, and assume that the property holds at its left end. For notational simplicity we will assume that $\sigma^{(j)}$ is the identity, so that on this interval $\pi_i(t)=\pi_{(i)}(t)$. Introduce the notation
$$
\delta_i(t):=\pi_i(t)-\nu_{i\wedge (k+1)}(t),\quad i\in[N].
$$
One has the following time derivatives
$$
\begin{array}{lll}
\frac{d}{dt}\pi_i&=\sum_{j\in[k], j\sim i}(\pi_j-\pi_i)+\sum_{j\notin[k],j\sim i}(\pi_j-\pi_i),&i\in[N],
\\
\frac{d}{dt}\nu_i&=\sum_{j\in [k], j\sim i}(\nu_j-\nu_i)+\sum_{j\notin [k],j\sim i}(\nu_{k+1}-\nu_i),& i\in[k],\\
\frac{d}{dt}\nu_{k+1}&=\sum_{i\notin [k]}\sum_{j\in [k],j\sim i}(\nu_j-\nu_{k+1}).&
\end{array}
$$
By the previous display one has for $i\in [k]$:
\begin{equation}\label{eq:1.1}
\frac{d}{dt}\delta_i=\sum_{j\in [n], j\sim i}(\delta_j-\delta_i).
\end{equation}
Note that, because the values $\pi_i$ are sorted, for all $j\notin [k]$, $\pi_j-\pi_{k+1}\le 0$. This together with the expression for the time derivative of $\pi_{k+1}$ yield
$$
\frac{d}{dt}\pi_{k+1}\le \sum_{j\in[k], j\sim k+1}(\pi_j-\pi_{k+1}). 
$$
Thus
\begin{equation}\label{eq:1.2}
\begin{array}{lll}
\frac{d}{dt}\delta_{k+1}&\le& \sum_{j\in[k],j\sim k+1}(\pi_j-\pi_{k+1})\\
&&-\sum_{i\notin [k]}\sum_{j\in [k],j\sim i}(\nu_j-\nu_{k+1})\\
&=&\sum_{j\in[k],j\sim k+1}(\delta_j-\delta_{k+1})\\
&&-\sum_{i\notin[k+1]}\sum_{j\in[k],j\sim i}(\nu_j-\nu_{k+1}).
\end{array}
\end{equation}
Let us argue by contradiction, and assume that there exists $t\in\setR_+$ and $i\in[N]$ for which $\delta_i(t)>0$. Let $\delta(t):=\sup_{j\in[N]}\delta_j(t)$. As the $\pi_j$ are sorted in decreasing order, one also has $\delta(t)=\sup_{j\in[k+1]}\delta_j(t)$. 

Since the $\delta_j(t)$ are linear combinations of finitely many exponentials, we can then identify an interval $J=[a,b]$ such that on $J$, for some $i\in[k+1]$, $\delta(t)\equiv \delta_i(t)$, and moreover $\delta(a)=0$, $\delta(t)>0$, $t\in(a,b]$. 

Assume that $i\in[k]$. From expression \eqref{eq:1.1}, we see that on $J$, $\frac{d}{dt}\delta=\frac{d}{dt}\delta_i\le 0$. This contradicts the fact that $\delta>0$ on $(a,b]$. 

Assume then that $i=k+1$. Then on $J$ one has, for all $j\in[k]$, as the $\pi_j$ are sorted,
$$
\nu_{k+1}=\pi_{k+1}-\delta_{k+1}\le \pi_{k+1}\le \pi_j =\nu_j+\delta_j\le \nu_j+\delta_{k+1}.
$$
Thus for all $j\in [k]$, $\nu_{k+1}-\nu_j\le \delta_{k+1}$. 
It then follows from \eqref{eq:1.2} that 
$$
\frac{d}{dt}\delta_{k+1}\le 0 + \alpha \delta_{k+1},
$$
where $\alpha=\sum_{i\notin [k+1]}\sum_{j\in[k],j\sim i}(1)$. Gronwall's lemma then implies that $\delta_{k+1}\le 0$ on $J$, a contradiction.
\end{proof}
\begin{remark}
When we move from interval $I^{(j)}$ to $I^{(j+1)}$ one can check that the meaning of distribution $\nu$ is preserved: we may change the permutation sorting the entries $\pi_i$, which results in a change in the graph used to define the evolution of $\nu$, but while the vertex to which $\nu_i$ refers may change, in that case the corresponding mass does not change. 
\end{remark}
\begin{lemma}\label{lemme_2}
Given a graph $G$ on vertex set $[N]$ with maximal degree $\Delta$ and for fixed  $k<n$,  associated isoperimetric constant $\phi_k(G)$, consider the graph $G'$ obtained by collapsing $N-k$ nodes into a single node as previously described. Then the resulting Laplacian matrix $L$ has spectral gap at least $\lambda_2\ge \lambda_2^*$.
\end{lemma}
\begin{proof}
Without loss of generality we assume nodes $k+1,\ldots, N$ of $G$ have been collapsed into node $k+1$ of $G'$. Let $f$ be an eigenvector of $L$ associated with its second smallest eigenvalue $\lambda_2$. We can always choose $f$ such that $f_{k+1}\le 0$. 

Define $g_v=\max(f_v,0)$, $v\in [k+1]$, and thus $g_{k+1}=0$. Let $W=\{v\in [k+1]: f_v>0\}$. One has
$$
\begin{array}{lll}
\lambda_2\sum_{u\in W}f_u^2&=&\sum_{u\in W}(L f)_u f_u\\
&=&\sum_{u\in W}\left[d_u f_u -\sum_{v\in [k+1]}a_{uv} f_v\right] f_u\\
&=&\sum_{u\in W}\sum_{v\in [k+1]}a_{uv}[f_u -f_v]f_u\\
&=&\sum_{u\in W}\sum_{v\in W}a_{uv}(f_u-f_v)f_u\\
&&+\sum_{u\in W}\sum_{v\notin W}a_{uv}(f_u-f_v)f_u\\
&\ge& \sum_{u\in W}\sum_{v\in W}a_{uv}(f_u-f_v)f_u\\
&&+\sum_{u\in W}\sum_{v\notin W} a_{uv}f_u^2\\
&=&\langle Lg, g\rangle
\end{array}
$$
Thus
$$
\lambda_2\ge \frac{\langle Lg, g\rangle}{\langle g, g\rangle}=:K.
$$
On the other hand, 
$$
\begin{array}{lll}
\sum_{(uv)\in E}a_{uv}(g_u+g_v)^2&=&2\sum_{(uv)\in E}a_{uv}(g_u^2+g_v^2)\\
&&-\sum_{(uv)\in E}a_{uv}(g_u-g_v)^2\\
&\le& 2\sum_{v\in V}d_v g_v^2\\
&\le & 2\Delta \langle g, g\rangle,
\end{array}
$$
where we have used the fact that $g_{k+1}=0$ to upper bound each product $d_v g_v^2$ by $\Delta g_v^2$. 

By Cauchy-Schwarz inequality, 
$$
\left( \sum_{(uv)\in E}a_{uv}|g_u^2-g_v^2|   \right)^2\le \left( \sum_{(uv)\in E}a_{uv}(g_u-g_v)^2 \right)  \left( \sum_{(uv)\in E}a_{uv}(g_u+g_v)^2\right)  .
$$
Combined, these bounds give
$$
\begin{array}{lll}
K&=&\frac{\left( \sum_{(uv)\in E}a_{uv}(g_u-g_v)^2 \right)  \left( \sum_{(uv)\in E}a_{uv}(g_u+g_v)^2\right)}{\langle g, g\rangle \sum_{(uv)\in E}a_{uv}(g_u+g_v)^2}\\
&\ge & \frac{\left( \sum_{(uv)\in E}a_{uv}|g_u^2-g_v^2|   \right)^2}{2\Delta\langle g, g\rangle^2}\cdot
\end{array}
$$
Let $0=t_0<t_1\cdots<t_m$ be the distinct values taken by the $g_v$. For $i=0,\ldots,m$, let $V_i:=\{v\in V: g_v\ge t_i\}$. Thus for $i>0$, $(k+1)\notin V_i$. Let 
$$
\begin{array}{lll}
M&:=&\sum_{(uv)\in E}a_{uv}|g_u^2-g_v^2|\\
&=&\sum_{i=1}^m \sum_{(uv)\in E, g_v<g_u=t_i} a_{uv}(g_u^2-g_v^2)\\
&=&\sum_{i=1}^m \sum_{u:g_u=t_i}\sum_{v:g_v=t_j, j<i}a_{uv}(t_i^2-t_{i-1}^2+\cdots -t_{j+1}^2+t_{j+1}^2-t_j^2)\\
&=&\sum_{i=1}^m\sum_{u\in V_i}\sum_{v\notin V_i}a_{uv}(t_i^2-t_{i-1}^2)\\
&=&\sum_{i=1}^me(V_i, \overline{V}_i)(t_i^2-t_{i-1}^2)\\
&\ge& \phi_k(G)\sum_{i=1}^m|V_i|(t_i^2-t_{i-1}^2)\\
&=&\phi_k(G)\sum_{i=1}^m t_i^2(|V_i|-|V_{i+1}|)\\
&=&\phi_k(G)\langle g, g\rangle.
\end{array}
$$
Combined, these results yield 
$$
\lambda_2\ge K\ge \frac{\left(\phi_k(G)\langle g, g\rangle   \right)^2}{2\Delta\langle g, g\rangle^2}=\lambda^*_2.
$$

\end{proof}


\subparagraph*{Acknowledgements.}
The authors would like to thank George Giakkoupis for bringing the problem studied in this paper to their attention, and to acknowledge stimulating discussions on the topic with both George Giakkoupis and Marc Lelarge. 




\bibliography{bib}

\begin{thebibliography}{10}

\bibitem{Albert_statisticalmechanics}
Réka Albert and Albert lászló Barabási.
\newblock Statistical mechanics of complex networks.
\newblock {\em Rev. Mod. Phys}, page 2002.

\bibitem{zhu2015expanders}
Zeyuan Allen~Zhu, Aditya Bhaskara, Silvio Lattanzi, Vahab~S. Mirrokni, and
  Lorenzo Orecchia.
\newblock Expanders via local edge flips.
\newblock {\em CoRR}, abs/1510.07768, 2015.

\bibitem{berestycki}
N.~Berestycki.
\newblock Mixing times of markov chains: Techniques and examples.
\newblock {\em Alea-Latin American Journal of Probability and Mathematical
  Statistics}, 2016.
\newblock URL:
  \url{http://www.statslab.cam.ac.uk/~beresty/teach/Mixing/mixing3.pdf}.

\bibitem{bhamidi2011mixing}
Shankar Bhamidi, Guy Bresler, and Allan Sly.
\newblock Mixing time of exponential random graphs.
\newblock {\em The Annals of Applied Probability}, 21(6):2146--2170, 12 2011.

\bibitem{Bollobas98a}
Bela Bollobas.
\newblock {\em Modern Graph Theory}.
\newblock Springer, 1998.

\bibitem{borcea}
J.~Borcea, P.~Brändén, and T.M. Liggett.
\newblock Negative dependence and the geometry of polynomials.
\newblock {\em Journal of the American Mathematical Society}, 22(2):521--567,
  2009.

\bibitem{Cooper:2009:FMC:1582716.1582742}
Colin Cooper, Martin Dyer, and Andrew~J. Handley.
\newblock The flip markov chain and a randomising p2p protocol.
\newblock In {\em Proceedings of the 28th ACM Symposium on Principles of
  Distributed Computing}, PODC '09, pages 141--150, New York, NY, USA, 2009.
  ACM.
\newblock URL: \url{http://doi.acm.org/10.1145/1582716.1582742}, \href
  {http://dx.doi.org/10.1145/1582716.1582742}
  {\path{doi:10.1145/1582716.1582742}}.

\bibitem{cooper:inria-00070627}
Colin Cooper, Ralf Klasing, and Tomasz Radzik.
\newblock {A randomized algorithm for the joining protocol in dynamic
  distributed networks}.
\newblock Technical Report RR-5376, {INRIA}, November 2004.
\newblock URL: \url{https://hal.inria.fr/inria-00070627}.

\bibitem{dubhashi}
D~Dubhashi and D.~Ranjan.
\newblock Balls and bins: a study in negative dependence.
\newblock {\em Basic Research in Computer Science}, RS-96-25, 1996.

\bibitem{Durrett:07}
Rick Durrett.
\newblock {\em Random Graph Dynamics}.
\newblock Cambridge University Press, Cambridge, 2007.
\newblock URL: \url{http://www.math.cornell.edu/~durrett/RGD/RGD.html}.

\bibitem{Feder:2006:LSM:1170136.1170489}
Tomas Feder, Adam Guetz, Milena Mihail, and Amin Saberi.
\newblock A local switch markov chain on given degree graphs with application
  in connectivity of peer-to-peer networks.
\newblock In {\em Proceedings of the 47th Annual IEEE Symposium on Foundations
  of Computer Science}, FOCS '06, pages 69--76, Washington, DC, USA, 2006. IEEE
  Computer Society.
\newblock URL: \url{http://dx.doi.org/10.1109/FOCS.2006.5}, \href
  {http://dx.doi.org/10.1109/FOCS.2006.5} {\path{doi:10.1109/FOCS.2006.5}}.

\bibitem{Friedman:1989:SER:73007.73063}
J.~Friedman, J.~Kahn, and E.~Szemer{\'e}di.
\newblock On the second eigenvalue of random regular graphs.
\newblock In {\em Proceedings of the Twenty-first Annual ACM Symposium on
  Theory of Computing}, STOC '89, pages 587--598, New York, NY, USA, 1989. ACM.
\newblock URL: \url{http://doi.acm.org/10.1145/73007.73063}, \href
  {http://dx.doi.org/10.1145/73007.73063} {\path{doi:10.1145/73007.73063}}.

\bibitem{jerrum1989approximating}
M.~Jerrum and Alistair Sinclair.
\newblock Approximating the permanent.
\newblock {\em SIAM J. Comput.}, 18(6):1149--1178, dec 1989.

\bibitem{Jonasson_2012}
Johan Jonasson.
\newblock Mixing times for the interchange process.
\newblock {\em Alea-Latin American Journal of Probability and Mathematical
  Statistics}, 9(2):667--683, 2012.

\bibitem{KimMedard}
Minkyu Kim and Muriel Medard.
\newblock Robustness in large-scale random networks.
\newblock {\em INFOCOM}, 2004.

\bibitem{peres2009book}
David~Asher Levin, Yuval Peres, and Elizabeth~Lee Wilmer.
\newblock {\em Markov chains and mixing times}.
\newblock Providence, R.I. American Mathematical Society, 2009.
\newblock With a chapter on coupling from the past by James G. Propp and David
  B. Wilson.
\newblock URL: \url{http://opac.inria.fr/record=b1128575}.

\bibitem{Papagelis2015refining}
Manos Papagelis.
\newblock Refining social graph connectivity via shortcut edge addition.
\newblock {\em ACM Trans. Knowl. Discov. Data}, 10(2):12:1--12:35, oct 2015.

\bibitem{schweinsberg2002bound}
Jason Schweinsberg.
\newblock An o(n2) bound for the relaxation time of a markov chain on
  cladograms.
\newblock {\em Random Struct. Algorithms}, 20(1):59--70, 2002.

\bibitem{Shi99modelsof}
Lingsheng Shi and Nicholas Wormald.
\newblock Models of random regular graphs.
\newblock In {\em IN SURVEYS IN COMBINATORICS}, pages 239--298. University
  Press, 1999.

\end{thebibliography}


\end{document}